\documentclass[a4paper,UKenglish]{lipics-v2019}
\usepackage{microtype}%if unwanted, comment out or use option "draft"

\usepackage{amssymb,amsmath,amsfonts, amsthm}    %ams
\usepackage{verbatim}
\usepackage{tikz}
\usepackage[ruled, vlined]{algorithm2e}
\nolinenumbers
\newcommand{\R}{{\mathbb R}}   % reals
\newcommand{\Proba}{{\mathbb P}}

\newcommand{\M}{\text{\sc max}}
\newcommand{\m}{\text{\sc min}}

\newcommand{\ran}{\text{\sc ran}}

\newcommand{\opt}{\mathrm{opt}}

\newcommand{\VM}{V_\text{\sc max}}
\newcommand{\Vm}{V_\text{\sc min}}
\newcommand{\Vr}{V_\text{\sc ran}}
\newcommand{\Vs}{V_\text{\sc sink}}
\newcommand{\frc}{\text{\sc For}}

\newtheoremstyle{break}
{\topsep}{\topsep}%
{\itshape}{}%
{\bfseries}{}%
{\newline}{}%
\theoremstyle{break}

{\everymath{\displaystyle\everymath{}}\array}%
{\endarray}

\newcommand{\pros}[3]{{\mathbb P}^{#1}_{#2}\left(#3\right)}

\newcommand{\espsb}[3]{{\mathbb E}^{#1}_{#2}\left[#3\right]}  %bracket
\newcommand{\val}{\mbox{Val}}
\newcommand{\free}{\mbox{\it free}}
\newcommand{\pretotal}{{\cal P}}
\newcommand{\total}{{\cal T}}

\title{Solving Simple Stochastic Games with few Random Nodes faster using Bland's Rule}

\titlerunning{Solving SSG with few Random Nodes faster using Bland's Rule}%optional, please use if title is longer than one line

\author{David Auger}{DAVID laboratory, University of Versailles Saint-Quentin-en-Yvelines, France}{david.auger@uvsq.fr}{}{}%mandatory, please use full name; only 1 author per \author macro; first two parameters are mandatory, other parameters can be empty.

\author{Pierre Coucheney}{DAVID laboratory, University of Versailles Saint-Quentin-en-Yvelines, France}{pierre.coucheney@uvsq.fr}{}{}

\author{Yann Strozecki}{DAVID laboratory, University of Versailles Saint-Quentin-en-Yvelines, France}{yann.strozecki@uvsq.fr}{}{}

\authorrunning{D. Auger and P. Coucheney and Y. Strozecki}%mandatory. First: Use abbreviated first/middle names. Second (only in severe cases): Use first author plus 'et al.'

\Copyright{David Auger and Pierre Coucheney and Yann Strozecki}%mandatory, please use full first names. LIPIcs license is "CC-BY";  http://creativecommons.org/licenses/by/3.0/

%mandatory

\category{}%optional, e.g. invited paper

\relatedversion{}%optional, e.g. full version hosted on arXiv, HAL, or other respository/website

\supplement{}%optional, e.g. related research data, source code, ... hosted on a repository like zenodo, figshare, GitHub, ...

\funding{}%optional, to capture a funding statement, which applies to all authors. Please enter author specific funding statements as fifth argument of the \author macro.

\acknowledgements{} %{I want to thank my mom, my dad and the goblins.}%optional

\EventEditors{Rolf Niedermeier and Christophe Paul}
\EventNoEds{2}
\EventLongTitle{36th International Symposium on Theoretical Aspects of Computer Science (STACS 2019)}
\EventShortTitle{STACS 2019}
\EventAcronym{STACS}
\EventYear{2019}
\EventDate{March 13--16, 2019}
\EventLocation{Berlin, Germany}
\EventLogo{}
\SeriesVolume{126}
\ArticleNo{58} % "New number" (=<article-no>) goes here!
\begin{document}

\maketitle

\begin{abstract}The best algorithm so far for solving Simple Stochastic Games is
  Ludwig's randomized algorithm~\cite{ludwig1995subexponential} which works in
  expected $2^{O(\sqrt{n})}$ time. We first give a simpler iterative variant of
  this algorithm, using Bland's rule from the simplex algorithm, which uses exponentially less random bits than Ludwig's version.
Then, we show how to adapt this method to the algorithm of Gimbert and
Horn~\cite{gimbert2008simple} whose worst case complexity is $O(k!)$, where $k$ is the
number of random nodes. Our algorithm has an expected running time of
$2^{O(k)}$, and works for general random nodes with arbitrary outdegree and probability distribution on outgoing arcs.
\end{abstract}

\newpage

\section{Introduction}

A \emph{simple stochastic game}, SSG for short, is a two-player zero-sum game, a
turn-based version of \emph{stochastic games} introduced by Shapley \cite{shapley1953stochastic}. SSGs were introduced by Condon \cite{condon1992complexity} and provide a simple framework that allows to study  algorithmic complexity issues underlying reachability objectives. An SSG is played by moving a pebble on a graph. Some nodes are divided between players \m \, and \M: 
if the pebble reaches a node controlled by a player then she has to move the pebble along
an arc leading to another node. Some other nodes are ruled by chance, the pebble following
one outgoing arc according to some given probability distribution. Finally, there are sink nodes 
with a rational value, which is the gain that \M-player achieves when the pebble reaches this sink. 

Player \M 's objective is, given a starting node for the pebble, to maximize the expectation of her gain
against any strategy of \m. One can show that it is enough to consider stationary deterministic strategies for both players \cite{condon1992complexity}. 
Though seemingly simple since the number of stationary deterministic strategies
is finite, the task of finding a pair of optimal strategies, or equivalently, of
computing the so-called \emph{optimal values} of nodes, is in complexity class $\mathrm{PPAD}$~\cite{etessami2010complexity} but not known to be in $\mathrm{P}$. 

Simple stochastic games are a powerful model since they can simulate many other
games such as parity games, mean or discounted payoff
games~\cite{andersson2009complexity,DBLP:journals/corr/abs-1106-1232}. However
these games are believed to be simpler than SSGs and better algorithms are known for
them; in particular, parity game is in quasi-polynomial
time~\cite{calude2017deciding}. Stochastic versions of the previous games also
exist and are
computationally equivalent to SSGs~\cite{andersson2009complexity}.
Interestingly, SSGs have many application domains, for instance autonomous urban driving~\cite{chen2013synthesis}, smart energy management~\cite{chen2013automatic}, model checking of the modal $\mu$-calculus~\cite{stirling1999bisimulation}, etc.

There are some restrictions for SSGs for which the problem of
finding optimal strategies is tractable. If the game is acyclic, it can be
solved in linear time, and in polynomial time for almost acyclic games (few
cycles or small feedback arc sets)~\cite{auger2014finding}. If there is no
randomness, the game can be solved in almost linear
time~\cite{andersson2008deterministic}. Furthermore, Gimbert and Horn were the
first to extend this result by giving Fixed Parameter
Tractable (FPT) algorithms in the number of random
nodes~\cite{gimbert2008simple}. They indeed show that optimal strategies depend only on the ordering of the values of random nodes, and not on their actual values.
Using this idea, they devise two algorithms. The first one exhaustively
enumerates these orders until it finds  one that 
actually corresponds to optimal values. The second one is a strategy
improvement algorithm based on an iterative refinement of the orders. Both have
a complexity of $k!n^{O(1)}$, where $k$ is the number of random nodes. It has
been improved to $\sqrt{k!}n^{O(1)}$ expected time in~\cite{dai2009new}, by
randomly selecting a good strategy as a starting point for a strategy improvement algorithm.
In fact, as remarked in~\cite{chatterjee2009termination}, the distance between the values of two consecutive strategies in any strategy improvement algorithm depends on the number of random nodes.
Hence any SSG can be solved in time $4^k n^{O(1)}$ (in fact $\sqrt{6}^kn^{O(1)}$ using Lemma $1.1$ in \cite{auger2014finding}). The complexity has been further improved to $2^k n^{O(1)}$ in~\cite{ibsen2012solving}, by using a value iteration algorithm. 
Here a bit of caution is in order; in some papers, random nodes can have an
arbitrary outdegree and probability distribution on outgoing arcs, and in some
other they must be binary with uniform distribution. In the former case, if we
denote by $p$ the bit-size of the largest probability distribution on a random node, the first two cited algorithms have a complexity of $p\cdot k!$ and $p \cdot \sqrt{k!}$. On the other hand, the two algorithms with an exponential complexity in $k$ have an exponential dependency on $p$ when adapted to this context.

Without the previous restrictions, only algorithms running in exponential time are known.
Most of them are strategy improvement algorithms, which produce a sequence of
strategies of increasing values. These algorithms, such as the classical
Hoffman-Karp~\cite{hoffman1966nonterminating} algorithm, rely on the switch
operation, which by a local best-response, produces a strategy with better
value. Several ways of choosing the nodes that are switched have been
proposed~\cite{tripathi2011strategy}, which can be compared to the rules
for pivot selection for the simplex algorithm in linear programming. Though efficient in
practice, these algorithms fail to run in polynomial time on a well designed input~\cite{friedmann2009exponential}.  
The best algorithm so far, proposed by Ludwig~\cite{ludwig1995subexponential,halman2007simple}, is also a strategy iteration algorithm using a randomized version of Bland's rule \cite{bland1977new} to choose a switch. It solves any SSG in expected time $2^{O(\sqrt{n})}$. The first analysis of this kind of algorithm is due to Kalai \cite{kalai1992subexponential} and it has been slightly improved recently~\cite{hansen2015improved}.

\paragraph*{Our contributions}

In Sec.~\ref{sec:ludwig}, we present an iterative variant of Ludwig's
recursive algorithm which uses less random bits. In the rest of the
paper we adapt the idea of this algorithm to carefully enumerate orders
of random nodes in an SSG. First, in Sec.~\ref{sec:main_alg}, we
present a pivot operation yielding a strategy improvement
algorithm, which improves the one of~\cite{gimbert2008simple}. 
This pivot operation comes from a randomized dichotomy on
all orders that we explain in details in Sec.~\ref{sec:analysis}, using an auxilliary game
similar to the one of~\cite{dai2009new}.
We prove that our algorithm finds the optimal strategies in expected
time polynomial in $2^k$ and $p$, where $k$ is the number of random
nodes and $p$ is the maximum bit-length of a distribution on a random
node, answering positively a question of Ibsen-Jensen and
Miltersen~\cite{ibsen2012solving}.

%ancienne version
% In Sec.~\ref{sec:ludwig}, we present an iterative variant of Ludwig's recursive algorithm which  is simpler and 
% uses less random bits. In Sec.~\ref{sec:main_alg}, we explain how to associate to each order a forcing strategy as in Gimbert and Horn~\cite{gimbert2008simple} and a value derived from an auxiliary game. We present a pivot operation which increases the value we have defined and use it to build
%  a new randomized strategy iteration algorithm inspired by Ludwig's algorithm.
% We prove in Sec.~\ref{sec:analysis}, that our algorithm finds the optimal strategies in expected time polynomial in $2^k$ and $p$, where $k$ is the number of random nodes and $p$ is the maximum bit-length of a distribution on a random node, answering positively a question of Ibsen-Jensen and Miltersen~\cite{ibsen2012solving}. 

\section{Definitions and classic results on simple stochastic games}

We here review definitions and results related to SSGs. We only
sketch what we need and refer to longer expositions such as
\cite{condon1992complexity, tripathi2011strategy} for more details.

\begin{definition}[SSG]
  A simple stochastic game (SSG) is defined by a directed graph $G=(V,A)$,
  where $V$ is the set of nodes and $A$ the set of arcs,
  together with a partition of $V$ in four parts $\VM$, $\Vm$,
  $\Vr$ and $\Vs$, whose elements are respectively called $\M$-nodes,
  $\m$-nodes, $\ran$-nodes (for {\it random}) and {\it sinks}.
  We require that every node $x \in V$ has
  outdegree
  at least one, while sink nodes have outdegree exactly $1$ consisting of a single loop on themselves.
  We also specify for every sink $x \in
  \Vs$ a value $\val(x)$ which is a rational number, and for every random node $x \in \Vr$
  a rational probability distribution $p(x)$ on the outneighbours
  of $x$.
\end{definition}

In the original version of Condon \cite{condon1992complexity}, all nodes
except sinks have outdegree exactly two, the probability distribution on every
$\ran$-node is $(\frac{1}2, \frac{1}2)$, and there are only two sinks, one with
value $0$ and another with value $1$. Here, we allow more than two sinks, with
general rational values, and also allow more than outdegree two for all 
non-sink nodes,
with an arbitrary probability distribution for $\ran$-nodes.
However, for Ludwig's Algorithm (see Algorithms \ref{alg:ludwig},
\ref{alg:ludwigdeterministe} and \ref{alg:ludwigrecursifpartial} in section
\ref{sec:ludwig}) we shall suppose that all $\M$-nodes have outdegree $2$ and call such games  {\it $\M$-binary}.

\paragraph*{Strategies and values}

We now define strategies, by which we mean stationary and pure strategies. This is
enough for our purpose and it  turns out to be sufficient for optimality, see~\cite{condon1992complexity}. 
Such strategies specify  the choice of a neighbour for every node of a given player. 

\begin{definition}[Strategy]
  A \emph{strategy} for player $\M$ is a map $\sigma$ from $\VM$ to $V$ such that
  $\forall x \in \VM,$ $(x,\sigma(x)) \in A.$
\end{definition}

Strategies for player $\m$ are defined analogously on $\m$-nodes and are usually denoted by $\tau$.

\begin{definition}[play]
  A \emph{play} is a sequence of nodes $x_0, x_1, x_2, \dots$ such that for all $t \geq 0$, 
  $(x_t, x_{t+1}) \in A.$
  Such a play is \emph{consistent} with strategies $\sigma$ and $\tau$, respectively for player $\M$ and player $\m$, if for all $t \geq 0$,
  $ x_t \in \VM \Rightarrow x_{t+1} = \sigma(x_t) $
  and
  $ x_t \in \Vm \Rightarrow x_{t+1} = \tau(x_t). $
\end{definition}

A couple of strategies $\sigma, \tau$ and an initial node $x_0 \in V$
define recursively a random play consistent with $\sigma, \tau$ by
setting $(i)\; x_{t+1} = \sigma(x_t)$ if $x_t \in \VM$,
$(ii)\; x_{t+1} = \tau(x_t)$ if $x_t \in \Vm$, $(iii)\; x_{t+1} = x_t$ if
$x_t \in \Vs$, and finally $(iv)\; x_{t+1}$ is one of the outneighbours
of $x_t$, randomly chosen independently of everything else according
to probability distribution $p(x)$,  if $x_t \in \Vr$.

Hence, %two strategies $\sigma, \tau$, together with an initial node $x_0$
this defines a probability measure $\Proba_{\sigma,\tau}^{x_0}$ on plays consistent with $\sigma, \tau$. Note that if a play contains a sink node $x_s$, then at every subsequent time the play stays in $x_s$. Such a play is said to \emph{reach} sink $x_s$. 
To every play $x_0, x_1, \dots$ we associate a value which is the value of the
sink reached by the play if any, and $0$ otherwise. If we denote by $X$ this
value, then $X$ is a random variable once two strategies and an
initial node $x$ are fixed. We are interested in the expected value of this quantity, which we call the value of a node $x \in V$ under strategies $\sigma, \tau$:  $\val_{\sigma,\tau} (x) = {\mathbb E}_{\sigma,\tau}^x \left( X \right)$ where ${\mathbb E}_{\sigma, \tau}^x$ is the expected value under probability $\Proba_{\sigma, \tau}^x$.

The goal of player $\M$ is to maximize this (expected) value, and the best he can ensure against a strategy $\tau$ is $\val_{*,\tau} (x) := \max_{\sigma} \val_{\sigma,\tau} (x)$ where the maximum is considered over all $\M$-strategies (which are in finite number).
Similarly, against $\sigma$ player \m  \ can ensure that the expected value is at most
$\val_{\sigma,*} (x) := \min_{\tau} \val_{\sigma,\tau} (x).$

Finally, the value of a node $x$ is  $\val_{*,*}(x) := \max_{\sigma} \val_{\sigma,*} (x) = \min_{\tau} \val_{*,\tau} (x).$
The fact that these two quantities are equal is nontrivial, and it can be found for instance in \cite{condon1992complexity}. A pair of strategies $\sigma^*, \tau^*$ such that, for all nodes $x$,
$\val_{\sigma^*,\tau^*} (x) = \val_{*,*}(x)$
always exists and these strategies are said to be \emph{optimal strategies}. It
is polynomial-time equivalent to compute optimal strategies or to compute the
values of all nodes in the game.

%\paragraph*{Stopping games}

\begin{definition}[Stopping SSG]
  An SSG is said to be \emph{stopping} if for every couple of strategies almost all plays eventually reach a sink node.
\end{definition}
Usually, this condition is required in order to ensure simple optimality conditions
(Thm. \ref{th:local_optimality_conditions3} below). 
Condon \cite{condon1992complexity} proved that every SSG $G$ can be reduced in
polynomial time to a stopping SSG $G'$ whose size is quadratic in the size of
$G$, and whose values almost remain the same.  The values of the new game are
close enough to recover the values of the original game. % , 
A problem for us is that squaring the size of the game does not behave well
relatively to precise complexity bounds.

However, in our case we need a milder condition. We call a \M-strategy $\sigma$ stopping  if, for any \m-strategy $\tau$, the random play consistent with $(\sigma, \tau)$ reaches a sink with probability one.

\begin{theorem}[Optimality conditions, \cite{condon1992complexity}]\label{th:local_optimality_conditions3}
  Let $G$ be an SSG, $\sigma$ a stopping \M-strategy and $\tau$ a \m-strategy. Then $(\sigma, \tau)$ are optimal strategies if and only if 
  \begin{itemize}
  \item for every $\displaystyle x \in \VM$, $\displaystyle \val_{\sigma,\tau}(x) = \max_{(x,y) \in A} \val_{\sigma,\tau} (y)$;
  \item for every $x \in \Vm$, $\displaystyle \val_{\sigma,\tau}(x) = \min_{(x,y) \in A} \val_{\sigma,\tau} (y)$.
  \end{itemize} 
\end{theorem} 

% conditions supprimees
%
%  \item for every $x \in \Vr$, $w(x) = \sum_{y \in V} p(x)[y] w(y)$
%    where $p(x)[y]$ is the probability of $y$ according to distribution $p(x)$; 
%
%  \item for every $x \in \Vs$, $w(x) = \val(x)$.

\paragraph*{Switches and strategy improvement}

Consider the usual partial order on real vectors indexed by $V$, i.e.
for $w_1, w_2 \in \R^V$, denote $w_1 \leq w_2$ if $w_1(x) \leq w_2(x)$
for all $x \in V$, and denote $w_1 < w_2$ if $w_1 \leq w_2$ and at least one
inequality is strict. For two $\M$-strategies $\sigma, \sigma'$, simply denote $\sigma \leq \sigma'$
(resp. $\sigma < \sigma'$) if $\val_{\sigma,*} \leq \val_{\sigma',*}$ (resp.
$\val_{\sigma,*} < \val_{\sigma',*}$).
Define a similar order on $\m$-strategies.

A switch, given a strategy, is the fact of changing this strategy at a node
(or a set of nodes) in order to obtain a new one.

\begin{definition}
  Let $\sigma, \sigma'$ be $\M$-strategies. We say that  $\sigma'$ is
  {\it a profitable switch} of $\sigma$ if
  for all $x \in \VM$, one has $\val_{\sigma,*}(\sigma'(x)) \geq  \val_{\sigma,*}(\sigma(x))$
    with this condition strict for at least one $\M$-node (such a node
    is said to be {\it switchable}).
\end{definition}

Indeed, the following result states that such a switch actually improves values

\begin{theorem}[\cite{condon1990algorithms}, \cite{tripathi2011strategy}] \label{th:switch_augmentation}
  If $\sigma'$ is a profitable switch of $\sigma$, then $\sigma' > \sigma$.
\end{theorem} 

Before ending this section, please note that Th.~\ref{th:local_optimality_conditions3}
can be restated in terms of nonexistence of switchable node. Hence, we have the
following result:

\begin{theorem} \label{th:noswicth}
  A stopping \M-strategy is optimal if and only if it has no switchable nodes.
\end{theorem}

For the last section, we require another form of switch.

\begin{theorem}[\cite{condon1990algorithms}, \cite{tripathi2011strategy}] \label{th:switch_augmentation2}
  Let $\sigma, \sigma'$ be stopping $\M$-strategies and $\tau,\tau'$ be $\m$-strategies
  such that
  for all $x \in \VM$, $\val_{\sigma,\tau}(\sigma'(x)) \geq  \val_{\sigma,\tau}(\sigma(x))$ 
    and  for all $x \in \Vm$, $ \val_{\sigma,\tau}(\tau'(x)) \geq  \val_{\sigma,\tau}(\tau(x))$
    with one of these conditions strict for at least one node.
    Then
    $\val_{\sigma', \tau'} > \val_{\sigma,\tau}.$
  \end{theorem}

\paragraph*{Orders}

For $k \geq 1$ consider the set of integers $[1,k] =  \{1,2,\cdots,k\}$ and let
$\total(k)$ denote the set of total orders on $[1,k]$.
For sake of clarity we view these orders as sets of couples $(i,j) \in [1,k]^2$
satisfying reflexivity, transitivity and antisymmetry.

If $t \in \total(k)$, it can also be described in  {\it ascending ordering} such as
$[x_1, x_2, \dots, x_k]$ where $(x_i,x_j) \in t$ if and only if $i \leq j$.
An {\it interval} in $t$ is a sequence of consecutive elements in ascending
ordering.
The rank of an element $x \in [1,k]$ is the number of elements that are lower of
equal to $x$ in $t$, i.e. it is $i$ if $x = x_i$ with notation above.

For lack of a better word, we define a {\it pretotal order} as an antisymmetric
and reflexive relation and denote by $\pretotal(k)$ the set of pretotal orders on
$[1,k]$. If $p \in \pretotal(k)$ and $(i,j) \not\in p$ is such that $p \cup
\{(i,j)\}$ is still antisymmetric, we denote simply by $p+(i,j)$ this new
pretotal order.

% 
% For a pretotal order $p$, if $i,j \in  [1,k]$ and neither $(i,j) \in p$ nor
% $(j,i) \in p$,  we say that $\{i,j\}$ is a {\it
%   free pair}. We denote $\free(p)$ the set of free pairs of $p$.
% 
If $t \in \total(k)$ and $v_1, v_2, \cdots,v_k$ are real numbers, we say that
the $v_i$'s are {\it nondecreasing} along $t$ if $ (i,j) \in t \Rightarrow v_i \leq v_j$.
Likewise, we say that $t$ is a \emph{nondecreasing order} for $v_1, v_2, \dots,
v_k$.

\begin{comment}
\subsubsection{extension of an order}
If $p \in \pretotal(k)$ and $\{i,j\} \in \free(p)$, let $ p +
(i,j) $ be the smallest partial order $p'$ extending $p$ with $(i,j) \in
p'$, i.e. the intersection of all partial orders containing $p$ and $(i,j)$.

\begin{lemma} \label{lem:extension}
  Let $p \in \partial(k)$ and $\{i,j\} \in \free(p)$. Then
  \[  p + (i,j)   = p \cup
    \{(i',j') : (i',i) \in p \text{ and } (j,j') \in p \}.\]
\end{lemma}

\begin{proof}
  Remember that an intersection of orders is an order.
  Let $A =  p \cup \{(i',j') : (i',i) \in p \text{ and } (j,j') \in p \}.$ 
  By transitivity we necessarily have $A \subset  \omega + (i,j)
   $. Now we just have to check that $A$
  is an order. todo.
\end{proof} 
\end{comment}

\section{Iterative formulation of Ludwig's algorithm} \label{sec:ludwig}

In this part, we suppose that $G$ is $\M$-binary. Hence, if a node $x$ is
switchable there is a single possibility for changing the strategy's choice at this node.
Let $switch(\sigma,x)$ denote the profitable switch obtained from $\sigma$
by switching $\sigma$ at node $x$.
\smallskip

\subsection{Bland's rule version}

In \cite{ludwig1995subexponential}, Ludwig mentions that his algorithm is a version of Bland's rule,
however he does not make it explicit and gives a recursive
definition. We formulate his algorithm iteratively (see Algorithm~\ref{alg:ludwig}), and show that instead of randomly choosing a node at every step,
we can choose a total order on nodes prior to the execution of the algorithm.
This version uses much less random bits : $O(n \log n)$ bits instead
of  $2^{O(\sqrt{n})}$ in average in Ludwig's version.
\begin{algorithm}
\DontPrintSemicolon
  \SetKwInOut{Input}{input}\SetKwInOut{Output}{output}
  \Input{$G$ $\max$-binary SSG, initial stopping $\M$-strategy $\sigma$.}
  \Output{an optimal $\M$-strategy}
  \BlankLine
  $\cdot$ Pick randomly and uniformly a total order $\Theta$ on $\M$-nodes\;
  \While{$\sigma$ is not optimal}{
    $\cdot$ compute the set of \emph{switchable} nodes for $\sigma$\;
    $\cdot$ let $x$ be the first switchable node in order $\Theta$\;
    $\cdot$ $\sigma \longleftarrow switch(\sigma, x)$\
  }
  \Return{$\sigma$}
  \caption{Bland's rule formulation for Ludwig's Algorithm} \label{alg:ludwig}
\end{algorithm}

By Theorems \ref{th:switch_augmentation} and \ref{th:noswicth} if
we proceed by switching Strategy $\sigma$ until there are no more switchable nodes,
we reach an optimal strategy in a finite number of steps.
The number of steps is at most the number of $\M$-strategies,
i.e. $2^{|\VM|}$. However, we have the following:

\begin{theorem} \label{thm:ludwigglobal}  
  The expected number of strategies considered by  Alg.~\ref{alg:ludwig} is at most $e^{2\sqrt{|\VM|}}$.
\end{theorem}

\subsection{Analysis of Algorithm~\ref{alg:ludwig}}

Our strategy to prove Theorem~\ref{thm:ludwigglobal}
is to reformulate Alg.~\ref{alg:ludwig} as a recursive algorithm  (see
Alg.~\ref{alg:ludwigrecursifpartial}), which is close to Ludwig's algorithm
in \cite{ludwig1995subexponential}. The proofs are quite similar to Ludwig's, with a bit of caution on the moments where random choices are made. 
In particular, we detail our strategy in this part since it will be helpful to
understand our results in section \ref{sec:main_alg} where the context
is more involved.

Stated as above, it is perhaps unclear how Alg.~\ref{alg:ludwig} has a recursive
structure. Too see this, consider an execution of Alg.~\ref{alg:ludwig}, and let
$x_1$ be the last $\M$-node in the order $\Theta$. In the beginning,
the current strategy $\sigma$ makes an initial choice $\sigma(x_1)$ on $x_1$,
which does not change until the first time when $x_1$ becomes switchable (if this
happens).
If $x_1$ is switched, then $\sigma(x_1)$ will then remain unchanged
until the end of this algorithm. Hence, once $\Theta$ is fixed, we can think of
this execution as two parts, where $\sigma(x_1)$ is fixed in each part. These can
then be decomposed as subparts where $\sigma(x_1)$ and $\sigma(x_2)$ are fixed (where $x_2$ is
the second-to-last $\M$-node in order $\Theta$), and so on.

\subsubsection*{Generalization to partially fixed strategies}

To formalize the discussion above, we give a
generalization which can be applied to the case where $\sigma(x)$
is fixed for some vertices in a given set $F$ (see Alg.~\ref{alg:ludwigdeterministe}).

In the following, if $F$ is a set of $\M$-nodes and $\sigma$ is a $\M$-strategy,
a {\it $(\sigma,F)$-compatible strategy} is any $\M$-strategy $\sigma'$ such
that $\forall x \in F$, $\sigma'(x) = \sigma(x).$
For $F$ and $\sigma$ fixed, there is always a $(\sigma,F)$-strategy that is
better than all others.
It can be obtained by solving the game where any $x \in F$ is replaced by a
random node with a probability $1$ to go to $\sigma(x)$.
We call such a $(\sigma,F)$-compatible strategy {\it optimal} and we denote it
by $\opt(\sigma, F)$.
In particular, an optimal $(\sigma,\emptyset)$-strategy is an optimal strategy
for $G$, whereas $\sigma$ is the only $(\sigma,\VM)$-compatible strategy.
\begin{algorithm}
\DontPrintSemicolon
  \SetKwInOut{Input}{input}
  \SetKwInOut{Output}{output}
  \Input{$G$ $\M$-binary SSG, total order $\Theta$ on $\VM$, subset $F \subset
    \VM$, initial $\M$-strategy $\sigma = \sigma_0$.}
  \Output{a $(\sigma,F)$-compatible optimal $\M$-strategy $\opt(\sigma, F)$.}
  \BlankLine \While{$\sigma$ is not an optimal $(\sigma_0,F)$-compatible
    strategy}{$\cdot$ compute the set of \emph{switchable} nodes for $\sigma$\;
    $\cdot$ let $v$ be the first switchable node in order $\Theta$ which is not in $F$\;
    $\cdot$ $\sigma \longleftarrow switch(\sigma, v)$\
  }
  \Return{$\sigma$}
  \caption{Iterative formulation for Ludwig's Algorithm with partial
    strategies} \label{alg:ludwigdeterministe} 
\end{algorithm} 
% ,
%     order $\Theta$ being fixed

\subsubsection*{Recursive reformulation}

Finally, we give a recursive version of Alg.~\ref{alg:ludwigdeterministe} (see
Alg.~\ref{alg:ludwigrecursifpartial}) which we use to derive the bound. The
equivalence between these two algorithms should be clear by the previous explanations.

\begin{algorithm} 
\DontPrintSemicolon
  \SetKwInOut{Input}{input}\SetKwInOut{Output}{output}
  \Input{$G$ $\M$-binary SSG, total order $\Theta$ on $\VM$, subset $F \subset \VM$, initial $\M$-strategy $\sigma_0$.}
  \Output{a $(\sigma,F)$-compatible optimal $\M$-strategy $\opt(\sigma, F)$.}
  \BlankLine
  \lIf{$F==\VM$}{return $\sigma$}
  
  \Else{
    $\cdot$ Let $v_0$ be the last node not in $F$ according to order $\Theta$\;
    $\cdot$ Recursively compute $\sigma_1$, an optimal $(\sigma_0,F\cup\{v_0\})$-compatible
    strategy\;
    \lIf{$\sigma_1$ is an optimal $(\sigma_0,F)$-compatible strategy}{\Return{$\sigma_1$}}
    
    \Else{
      $\cdot$ Let $\sigma_2 \longleftarrow switch(\sigma_1,v_0)$\;
      $\cdot$ Recursively compute $\sigma^*$, an optimal $(\sigma_2,F\cup\{v_0\})$-compatible
      strategy\;
      \Return{$\sigma^*$}
    }
  }
  
  \caption{Recursive formulation for Ludwig's Algorithm with partial
    strategies} \label{alg:ludwigrecursifpartial} 
\end{algorithm} 
% , order $\Theta$ being fixed

\subsubsection*{Evaluating the number of switches}

Let $f^\Theta(\sigma, F)$ be the total number of 
switches performed by Algorithm \ref{alg:ludwigrecursifpartial} on input $\sigma, \Theta, F$.
We consider for the following lemma an execution of this algorithm.

\begin{lemma}\label{ludwig-lemme1}

  Let $\sigma_0$ be the initial strategy and
  $v_0$ be the last node  which is not in $F$, according to order $\Theta$. Define
  $B \subset \VM \setminus F$ to be the set of nodes $v$ such that
  $\opt(\sigma_0, F \cup \{v\} ) \not> \opt(\sigma_0, F \cup \{v_0\}).$
  Then
  $f^\Theta(\sigma_0, F) \leq f^\Theta(\sigma_0, F \cup \{v_0\}) + 1 + 
    f^\Theta( \sigma_2,   F \cup B \}),$
  where $\sigma_2$ is $\opt(\sigma_0, F \cup \{v_0\}) = \sigma_1$,
  switched at $v_0$.
\end{lemma}
\begin{proof}
  By design of the algorithm, nodes of $F$ are never switched.    
  If $v_0$ is never switched, then we have \[ f^\Theta(\sigma_0, F) = f^\Theta(\sigma_0,
    F \cup \{v_0\})\] hence the result is true.

  Suppose from now on that $v_0$ is switched during the execution. Then, it is switched only once and we can divide the computation of $\opt(\sigma_0,F)$ in two parts:

  \begin{itemize}
  \item in a first part the algorithm computes $\opt(\sigma_0, F \cup \{v_0\})$,
    and this last strategy is switched at $v_0$,  hence obtaining
    $\sigma_2$;
  \item then in a second part the algorithm computes $\opt(\sigma_2, F)$.
  \end{itemize} 
  Hence, we have in the case that $v_0$ is switched,
  \[ f^\Theta(\sigma_0, F) = f^\Theta(\sigma_0, F \cup \{v_0\}) + 1 + 
    f^\Theta( \sigma_2,   F ).\]

  It remains to see that in the second part, nodes from $B$ will never be
  switched, so that $f^\Theta( \sigma_2,   F ) =   f^\Theta( \sigma_2,   F
  \cup B )$, hence the result.

  Let $\sigma'$ be a strategy obtained during the second part of the algorithm
  (after $v_0$ is switched) such that $\sigma'(v) = \sigma_0(v)$ for a
  $v \in \VM \setminus F$. On the one hand, by definition of $\opt$
  we have 
  \[\sigma' \leq \opt(\sigma', F \cup \{v\}) = \opt(\sigma_0, F \cup \{v\}).\]
  On the other hand, since $\sigma'$ is obtained after $\sigma_2$, then
  \[ \opt(\sigma_0, F \cup \{v_0\}) < \sigma_2 \leq \sigma'. \] 
  By transitivity we see that
  \[ \opt(\sigma_0, F \cup \{v_0\} ) < \opt(\sigma_0, F \cup \{v\}). \]
  hence $v \not\in B$.

  Therefore, in the second part of the algorithm, all strategies $\sigma'$
  satisfy $\sigma'(v) \neq \sigma_0(v)$ for all $v \in B$, hence nodes in $B$
  have been switched in the first part and never will be in the second.
\end{proof} 

Now, let us denote $\Phi(n) = \sup_{G,\sigma} \espsb{\Theta}{}{ f^\Theta_G(\sigma, \emptyset)}$
where the supremum is considered over all SSG $G$ with $n$ $\M$-nodes and all
$\M$-strategies $\sigma$ in $G$. The average is considered over
all possible prior choices of order $\Theta$, the rest of
the algorithm being deterministic.

% Hence,  $\Phi(n)$ is an upper bound on the average number of switches that Algorithm \ref{alg:ludwig} will
% perform on any SSG $G$ with $n$ $\M$-nodes for any strategy $\sigma$.

\begin{lemma} \label{lem:ineq1}
  For all $n \geq 1$, $\Phi(n) \leq \Phi(n-1) + 1 + \frac1n \sum_{i=0}^{n-1} \Phi(i).$
\end{lemma} 

We shall first need the following result.

\begin{lemma}\label{lem:poset1}
  Let $(X, \leq)$ be a partially ordered set and define for any $x \in X$ 
  $a(x) = |\{ y \in X : y \not> x \}|$, i.e. the number of elements that are not
  greater than $x$.  Then for all $0 \leq i \leq |X|$, we have
  \[ | \{ x : a(x) \leq i \} | \leq i .\]
\end{lemma}

\begin{proof}
  Fix $0 \leq i \leq |X|$ and consider the set $P \subset X$ of $x \in X$ with
  $a(x) \leq i$. Let $x_0$ be maximal among elements of $P$. Since there are at
  least $n-i$ elements in $X$ 
  that are strictly greater than $x_0$, and that these elements are not in $P$ by
  maximality of $x_0$, we have $|P| + n-i \leq n$ i.e. $|P| \leq i$.
\end{proof}

\begin{proof} 
  First, denote $\Phi(n,k)$ for $n \geq 1$ and $0 \leq k \leq n$ 
  \[\Phi(n,k) = \sup_{G,\sigma,H} \espsb{\Theta}{}{ f^\Theta_G(\sigma, H)} \]
  where the supremum is considered over all SSG $G$ with $n$ $\M$-nodes,
  subsets $H \subset \VM$ of size $k$ and all $\M$-strategies $\sigma$.

  Consider $G, H$ and $\sigma$ fixed, with $|\VM|=n$ and $|H|=k$. Using notation of Lemma \ref{ludwig-lemme1},
  we have 
  \[ \espsb{\Theta}{}{f^\Theta(\sigma, H)} \leq \espsb{\Theta}{}{f^\Theta(\sigma,
      H \cup \{v_0^t\})} + 1 + \espsb{\Theta}{}{ f^\Theta( \sigma', H \cup  B^t \})
    }.\]

  Here, we denote $v_0$ and $B$ by $v_0^t$ and $B^t$ to stress the fact
  that these are random variables
  depending on $\Theta$, whereas
  everything else ( i.e. $G, F, \sigma$) is fixed.

  First, since for all $v \not\in H$,
  \[ \espsb{\Theta}{}{f^\Theta(\sigma,
      H \cup \{v\})} \leq \Phi(n,k+1) \]
  we have \[\espsb{\Theta}{}{f^\Theta(\sigma,
      H \cup \{v_0^t\})} \leq \Phi(n,k+1). \]
  Now,
  \begin{align*}
    \espsb{\Theta}{}{ f^\Theta( \sigma', H \cup  B^t \})} &= \sum_{i=1}^{n-k} \pros{t}{}{|B^t|=i}  \cdot \espsb{\Theta}{}{ f^\Theta( \sigma', H \cup  B^t \})\ \big\vert \ |B^t|=i} \\
                                                &\leq  \sum_{i=1}^{n-k}   \pros{t}{}{|B^t|=i}  \cdot \Phi(n,k+i).
  \end{align*} 

  The sum on the right can easily be rewritten as

  \[ \sum_{i=1}^{n-k}   \pros{t}{}{|B^t| \leq i}  \cdot \big(\Phi(n,k+i) -
    \Phi(n,k+i+1)\big) 
  \]
  where we defined for convenience $\Phi(n,n+1) = 0$, and used the fact that
  $|B^t| \geq 1$.

  Using now Lemma \ref{lem:poset1} on the set of strategies $\opt(\sigma,H \cup
  \{v\})$ for $v \not\in H$, we see that
  \[  \pros{t}{}{|B^t| \leq i} \leq \frac{i}{n-k}\]
  since $v_0$ is uniformly chosen in $\VM \setminus H$.

  So, since $\Phi(n,k+i) \geq \Phi(n,k+i+1)$, we deduce that
  \begin{align*}
    \espsb{\Theta}{}{ f^\Theta( \sigma', H \cup  B^t \})} &\leq 
                                                  \sum_{i=1}^{n-k}   \frac{i}{n-k}  \cdot \big(\Phi(n,k+i) - \Phi(n,k+i+1)\big) \\
                                                &= \sum_{i=1}^{n-k}   \frac{1}{n-k}  \cdot \Phi(n,k+i).
  \end{align*} 
\end{proof}

In order to conclude and prove Theorem~\ref{thm:ludwigglobal}, we now just have
to infer the bound for sequences satisfying the conclusion of Lemma \ref{lem:ineq1}.

\begin{lemma}[Lemma $9$ of~\cite{ludwig1995subexponential}] \label{lem:borneF}
  Let $\Phi(n)$ be such that $\Phi(0)=0$ and for all $n \geq 1$,\\
  $\Phi(n) \leq \Phi(n-1) + 1 + \frac1n \sum_{i=0}^{n-1} \Phi(i).$
  Then for all $n \geq 0$, $\Phi(n) \leq e^{2\sqrt{n}}.$
\end{lemma}

%%%%%%%%%%%%%%%%%%%%%%%%%%%%%%%%%%%%%%%%%%%%%%%%%%%%%%%%%%%%%%%%%%%%%%%%%%%%%% 
%%%% LUDWIG ON RANDOM NODES
%%%%%%%%%%%%%%%%%%%%%%%%%%%%%%%%%%%%%%%%%%%%%%%%%%%%%%%%%%%%%%%%%%%%%%%%%%%%%% 

\section{Simple stochastic games with few random nodes}\label{sec:main_alg}

The idea that in an SSG, the optimal strategies depend only
on the ordering of the values of $\ran$-nodes, and not on their actual values, has been introduced by Gimbert and Horn in \cite{gimbert2008simple}. Their main idea is that, if one gives an ordering $r_1 r_2
\cdots r_k$ of $\ran$-nodes such that $\val_{*,*}(r_i)$ is nondecreasing with $i$,
then $\M$ will try to reach a node $r_i$ with $i$ as high as
possible, whereas $\m$ will try to minimize this index; this idea
is hereafter formalized by the notion of forcing sets and forcing strategies
(sec. \ref{sub:modified}). Gimbert and Horn use this
fact to derive an algorithm that will enumerate
all possible orders on $\ran$-nodes an will identify
one with the property mentionned above, yiedling the optimal strategies
and values for $G.$

The algorithm that we describe and analyse in the rest of this paper (Alg.~\ref{alg:iterative_random_ludwig}) uses
the same principle, but iterates through orders in a special
way, similarly to the iteration through strategies made
by Ludwig's algorithms (see sec. \ref{sec:ludwig}). We will derive
a similar bound for the average number of iterations of this randomized
algorithm.
Hence, our main algorithm is 
still a variation on Bland's rule for pivot selection. The difficulty here
does not lie in the proof of the bound, but in the description of the technique
used to iterate on orders. 

In \cite{gimbert2008simple}, the game remains the same during the execution of
the algorithm, but we proceed
differently: 
\begin{itemize}
  \item in section \ref{sub:modified}, we describe how to associate to every total order $t \in \total(k)$
a new SSG $G[t]$, and we show that this game can be solved in
polynomial time.
\item in section \ref{sub:pivot}, we prove that there is an optimal order $t^*
\in \total(k)$
such that the optimal values of $G[t^*]$ give directly the optimal values of $G$; it is also the
order that maximises values of $G[t]$ among all total orders $t$.
If an order $t$ is not optimal, we describe a {\it pivot} operation yielding from
$t$ a new order $t'$ such that the optimal values of $G[t']$ improve
those of $G[t]$.
\item the proof of the bound will be derived in section \ref{sec:analysis}.
\end{itemize}

\subsection{Modified game and forcing strategies} \label{sub:modified}

We need to assume that the games we consider enjoy some basic properties in order to describe our algorithm without considering too many special cases.

\begin{definition}
  An SSG is in \emph{canonical form} (CF) if $\M$ has a stopping strategy and only $\ran$-nodes can have an outgoing arc to a sink.
% :
%   \begin{enumerate}[(i)]
%   \item it is stopping ;
%   \item only $\ran$-nodes can have an outgoing arc to a sink.
%   \end{enumerate} 
\end{definition}

To ensure these conditions, one can first
in linear time find and remove all nodes from which $\m$ player can
force the game never to reach neither a sink node nor a $\ran$-node
(see e.g. \cite{andersson2008deterministic, condon1992complexity}).
These nodes have value $0$ and can as well be removed from the game.
Then, all probabilities on $\ran$-nodes are modified by giving them a
very small probability to go to a sink. One can prove as in
\cite{condon1992complexity} that values remain almost the same.
% We already discussed previously how to make the game stopping. As our
% algorithm is already quite involved, this condition simplifies exposure and proofs.
The second condition ensures that all $\M$ and $\m$ nodes have to reach
a $\ran$-node in order to reach a sink. It can be done by adding
a dummy random node before every sink.

% 
% \paragraph*{Modified game}

In all that follows we suppose that $G$ is an SSG in CF with random nodes $r_1,
r_2, \dots, r_k$. Let $t \in \total(k)$ be a total order on $[1,k]$. We define a game  $G[t]$ as follows (the same construction is presented in~\cite{dai2009new}).
Start with a copy of $G$. For every $1 \leq i \leq k$, add a $\m$-node denoted $i$ to $G[t]$, which we call {\it control node}; add an arc $(i,r_i) $; for every arc $(x,r_i) \in A$, remove this arc and add an arc $(x,i)$; finally, for every $(i,j) \in t$, $i \neq j$, add the arc $(i,j)$ to $G[t]$.

% \begin{itemize}
% \item add a $\m$-node denoted $i$ to $G[t]$, which we call {\it control node} ;
% \item add an arc $(i,r_i) $ ;
% \item for every arc $(x,r_i) \in A$, remove this arc and add an arc $(x,i)$;
% \item finally, for every $(i,j) \in t$, $i \neq j$, add the arc $(i,j)$ to $G[t]$.
% \end{itemize} 

So basically, every control node $i \in [1,k]$ intercepts all arcs entering in
$r_i$ (see Fig. \ref{fig:structure}), and has an arc to every other control node $j \in [1,k]$ which is greater than $i$ in $t$.
In the game $G[t]$, the set of sinks, $\M$-nodes and $\ran$-nodes remain the
same as in $G$, whereas the set of $\m$-nodes will be denoted $\Vm \cup [1,k]$,
where $\Vm$ is the set of $\m$-nodes in $G$. This allows us to directly identify
$\M$-strategies in $G[t]$ and in $G$, and to identify projections onto
$\Vm$ of $\m$-strategies in $G[t]$, to $\m$-strategies in $G$.

\begin{figure}[!h] 
%\hspace{-5.5em}
  \centering
  \caption{On the left, the structure of a game $G$ in canonical form, only
    random nodes can directly access sink nodes. On the right,
  the structure of $G[t]$.}\label{fig:structure}
  \begin{tikzpicture}[scale=0.6]
      \draw[rounded corners, line width=1pt] (-0.2,-0.1) rectangle (5.2,9.2) ;         
      
      \draw[rounded corners, line width=1pt] (0,0) rectangle (5,3.8) ;         
      \draw[rounded corners, line width=1pt] (0,4.2) rectangle (5,6.8) ;         
      \draw[rounded corners, line width=1pt] (0,7.2) rectangle (5,8.8) ;

      \node at (2.5,2) {$\VM \cup \Vm$};
      \node at (2.5,5.5) {$\Vr$};
      \node at (2.5,8) {$\Vs$};

      \draw[->, line width=2pt] (2.5,3.6) -- (2.5,4.4) ;
      \draw[->, line width=2pt] (1.5,3.6) -- (1.5,4.4) ;
      \draw[->, line width=2pt] (3.5,3.6) -- (3.5,4.4) ;

      \draw[->, line width=2pt] (1.5,6.6) -- (1.5,7.4) ;
      \draw[->, line width=2pt] (2.5,6.6) -- (2.5,7.4) ;
      \draw[->, line width=2pt] (3.5,6.6) -- (3.5,7.4) ;

      \draw[->, line width=2pt] (4.7,5.8) to[out = -30, in =30] (4.6,2.2) ;
      \draw[->, line width=2pt] (4.7,6.2) to[out = -30, in =30] (4.6,1.8) ;
 
    \begin{scope}[shift={(8,0)}]
      \draw[rounded corners, line width=1pt] (-0.2,-0.1) rectangle (5.2,9.2) ;         
      
      \draw[rounded corners, line width=1pt] (0,0) rectangle (5,2.8) ;         
      \draw[rounded corners, line width=1pt] (0,3.2) rectangle (5,4.8) ;         
      \draw[rounded corners, line width=1pt] (0,5.2) rectangle (5,6.8) ;         
      \draw[rounded corners, line width=1pt] (0,7.2) rectangle (5,8.8) ;

      \node at (2.5,1.5) {$\VM \cup \Vm$};
      \node at (2.5,4) {control nodes};
      \node at (2.5,6) {$\Vr$};
      \node at (2.5,8) {$\Vs$};

      \draw[->, line width=2pt] (2.5,2.6) -- (2.5,3.4) ;
      \draw[->, line width=2pt] (1.5,2.6) -- (1.5,3.4) ;
      \draw[->, line width=2pt] (3.5,2.6) -- (3.5,3.4) ;

      \draw[->, line width=2pt] (1.5,4.6) -- (1.5,5.4) ;
      \draw[->, line width=2pt] (2.5,4.6) -- (2.5,5.4) ;
      \draw[->, line width=2pt] (3.5,4.6) -- (3.5,5.4) ;

      \draw[->, line width=2pt] (1.5,6.6) -- (1.5,7.4) ;
      \draw[->, line width=2pt] (2.5,6.6) -- (2.5,7.4) ;
      \draw[->, line width=2pt] (3.5,6.6) -- (3.5,7.4) ;

      \draw[->, line width=2pt] (4.7,5.8) to[out = -30, in =30] (4.6,2.2) ;
      \draw[->, line width=2pt] (4.7,6.2) to[out = -30, in =30] (4.6,1.8) ;
    \end{scope} 
  \end{tikzpicture}
  \end{figure}
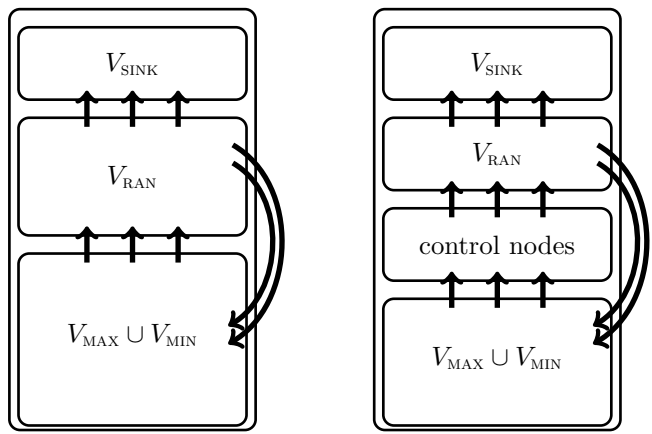

Now, suppose we remove first all sinks and
random nodes of $G[t]$, and then turn every control node $i$ into a
sink with a value equal to its rank in $t$. This transformation
clearly turns $G[t]$ into a game $G'$ without random nodes.

\begin{definition}[Forcing strategy]
  By identifying strategies in $G[t]$ and $G'$, we say that any
  optimal strategy for $\M$ or $\m$ in $G'$ is a \emph{ $t$-forcing
    strategy} of $G[t]$.
\end{definition} 

In $t$-forcing strategies, the players try to ensure the reaching of a control node as high
as possible for $\M$, and as low as possible for $\m$, in the order $t$.
We refer to \cite{andersson2008deterministic} and \cite{gimbert2008simple} for
more details about how one can compute these optimal strategies in linear time,
using the so-called {\it deterministic attractors}.

% Forcing strategies associated to an order $t$ will usually be denoted
% as $(\sigma_t, \tau_t)$. Note that these strategies are not
% necessarily unique but that the actual choice of a forcing strategy is
% in general irrelevant since they all give the same values to all nodes
% in the game.

\begin{definition}[Forcing set]
  For any control node $i \in [1,k]$, define the \emph{forcing set for
    $i$}, denoted $\frc[t](i)$, as the set of $\M$ and $\m$-nodes that
  reach $i$ if the game is played with a couple $(\sigma_t,\tau_t)$ of
  $t$-forcing strategies (forcing sets are independant of the choice
  of the strategies as long as they are $t$-forcing).
\end{definition}

% \inpierre{
% An obvious fact:
% \begin{lemma}\label{lem:forcingSet}
%   For any control node $i \in [1,k]$ consider $i^-(t)$ the set of control nodes that are lower than $i$ in the order $t$. If $i^-(t) = i^-(t')$, then $\frc[t](i) = \frc[t'](i)$.
% \end{lemma}

% Je crois pas que ce soit utile de le mettre ici. C'est utilisé dans la preuve du théorème 28.}

An example of an SSG turned into a modified SSG and of computation of forcing strategies is presented in
Fig.~\ref{fig.example}.
\begin{figure}[!h]
\hspace{-5.5em}
\begin{subfigure}[b]{0.3\textwidth}
  \centering
  \caption{}\label{fig.a}
      \begin{tikzpicture}[scale=0.6]
      \draw[rounded corners, line width=1pt] (-1,0) rectangle (9,9.1) ;         
      \tikzstyle{block} = [rectangle, draw, fill=blue!20, text width=2em, text centered, rounded corners, minimum height=2em]
      \tikzstyle{aver} = [circle, draw, fill=red!20, text centered, minimum height=2em]
      \tikzstyle{final} = [rectangle, draw, fill=yellow!30, text width=2.3em, text centered, rounded corners, minimum height=2em]
      \tikzstyle{etiquette}=[midway,fill=black!20]
      \tikzstyle{operation}=[->,>=latex, line width =1pt]
      
      \node[block] (M) at (3,2) {$M$} ;
      \node[block] (m) at (6,1) {$m$} ;

      \node[aver] (a) at (2,4) {$r_1$} ;
      \node[aver] (b) at (4,4) {$r_2$} ;
      \node[aver] (c) at (6,4) {$r_3$} ;

      \node[final] (0) at (0,7) {\tiny Val$=0$} ;
      \node[final] (5) at (4,7) {\tiny Val$=.5$} ;
      \node[final] (1) at (8,7) {\tiny Val$=1$} ;

      \draw[->,line width =1pt] (a) edge [bend left = 25, left] node {$.09$ }(0) ; 
      \draw[->,line width =1pt] (a) edge [bend left = 15, right] node {$.01$ }(1) ;
      \draw[->,line width =1pt, bend right = 80] (a) edge [left] node {$.9$ } (m) ;

      \draw[->,line width =1pt] (b) edge (5) ; 

      \draw[->,line width =1pt] (c) edge [ left] node{$.01$} (0) ; 
      \draw[->,line width =1pt] (c) edge [bend right = 20, left] node{$.09$} (1) ;
      \draw[->,line width =1pt] (c) edge[out=0,in=0, right] node{$.9$} (m) ;

      \draw[->,line width =1pt] (M) edge[bend right=15, left] (a) ; 
      \draw[->,line width =1pt] (M) edge[bend left=15, right] (b) ;
      
      \draw[->,line width =1pt] (m) edge[left] (M) ; 
      \draw[->,line width =1pt] (m) edge[bend left=15, right] (c) ;
    \end{tikzpicture}
\end{subfigure}
\hspace{5em}
\begin{subfigure}[b]{0.3\textwidth}
  \begin{center}
    \caption{}\label{fig.b}
    \begin{tikzpicture}[scale=0.6]
      \draw[rounded corners, line width=1pt] (-1,-1.1) rectangle (9,8) ;         
      \tikzstyle{block} = [rectangle, draw, fill=blue!20, text width=2em, text centered, rounded corners, minimum height=2em]
      \tikzstyle{aver} = [circle, draw, fill=red!20, text centered, minimum height=2em]
      \tikzstyle{final} = [rectangle, draw, fill=yellow!30, text width=2.3em, text centered, rounded corners, minimum height=2em]
      \tikzstyle{etiquette}=[midway,fill=black!20]
      \tikzstyle{operation}=[->,>=latex, line width =1pt]
      
      \node[block] (M) at (3.8,0.3) {$M$} ;
      \node[block] (m) at (6,0) {$m$} ;

      \node[block] (ma) at (2,2) {$1$} ;
      \node[block] (mb) at (4,2) {$2$} ;
      \node[block] (mc) at (6,2) {$3$} ;

      \node[aver] (a) at (1.5,4) {$r_1$} ;
      \node[aver] (b) at (4,4) {$r_2$} ;
      \node[aver] (c) at (6,4) {$r_3$} ;

      \node[final] (0) at (0,7) {\tiny Val$=0$} ;
      \node[final] (5) at (4,7) {\tiny Val$=.5$} ;
      \node[final] (1) at (8,7) {\tiny Val$=1$} ;

      \draw[->,line width =1pt] (a) edge [bend left = 25, left] node {$.09$ }(0) ; 
      \draw[->,line width =1pt] (a) edge [bend left = 15, right] node {$.01$ }(1) ;
      \draw[->,line width =1pt, bend right = 80] (a) edge [left] node {$.9$ } (m) ;

      \draw[->,line width =1pt] (b) edge (5) ; 

      \draw[->,line width =1pt] (c) edge [left] node{$.01$} (0) ; 
      \draw[->,line width =1pt] (c) edge [bend right = 20, left] node{$.09$} (1) ;
      \draw[->,line width =1pt] (c) edge[out=0,in=0, right] node{$.9$} (m) ;

      \draw[->,line width =1pt] (ma) edge (a) ; 
      \draw[->,line width =1pt] (mb) edge (b) ;
      \draw[->,line width =1pt] (mc) edge (c) ;

      \draw[->,line width =1pt] (M) edge[bend right=15] (ma) ; 
      \draw[->,line width =1pt] (M) edge[bend left=15] (mb) ;
      
      \draw[->,line width =1pt] (m) edge (M) ; 
      \draw[->,line width =1pt] (m) edge[bend left=15] (mc) ;
    \end{tikzpicture}
  \end{center}
\end{subfigure}
\hspace{5em}
\begin{subfigure}[b]{0.3\textwidth}
  \begin{center}
      \caption{}\label{fig.c}
    \begin{tikzpicture}[scale=0.6]
      \draw[rounded corners, line width=1pt] (-1,-1.1) rectangle (9,8) ;         
      \tikzstyle{block} = [rectangle, draw, fill=blue!20, text width=2em, text centered, rounded corners, minimum height=2em]
      \tikzstyle{aver} = [circle, draw, fill=red!20, text centered, minimum height=2em]
      \tikzstyle{final} = [rectangle, draw, fill=yellow!30, text width=2.3em, text centered, rounded corners, minimum height=2em]
      \tikzstyle{etiquette}=[midway,fill=black!20]
      \tikzstyle{operation}=[->,>=latex, line width =1pt]
      
      \node[block] (M) at (3.8,0.3) {$M$} ;
      \node[block] (m) at (6,0) {$m$} ;

      \node[block] (ma) at (2,2) {$1$} ;
      \node[block] (mb) at (4,2) {$2$} ;
      \node[block] (mc) at (6,2) {$3$} ;

      \node[aver] (a) at (1.5,4) {$r_1$} ;
      \node[aver] (b) at (4,4) {$r_2$} ;
      \node[aver] (c) at (6,4) {$r_3$} ;

      \node[final] (0) at (0,7) {\tiny Val$=0$} ;
      \node[final] (5) at (4,7) {\tiny Val$=.5$} ;
      \node[final] (1) at (8,7) {\tiny Val$=1$} ;

      \draw[->,line width =1pt] (a) edge [bend left = 25, left] node {$.09$ }(0) ; 
      \draw[->,line width =1pt] (a) edge [bend left = 15, right] node {$.01$ }(1) ;
      \draw[->,line width =1pt, bend right = 80] (a) edge [left] node {$.9$ } (m) ;

      \draw[->,line width =1pt] (b) edge (5) ; 

      \draw[->,line width =1pt] (c) edge [left] node{$.01$} (0) ; 
      \draw[->,line width =1pt] (c) edge [bend right = 20, left] node{$.09$} (1) ;
      \draw[->,line width =1pt] (c) edge[out=0,in=0, right] node{$.9$} (m) ;

      \draw[->,line width =1pt, dashed] (ma) edge (a) ; 
      \draw[->,line width =1pt, dashed] (mb) edge (b) ;
      \draw[->,line width =1pt] (mc) edge (c) ;

      \draw[->,line width =1pt] (M) edge[bend right=15] (ma) ; 
      \draw[->,line width =1pt, dashed] (M) edge[bend left=15] (mb) ;
      
      \draw[->,line width =1pt] (m) edge (M) ; 
      \draw[->,line width =1pt, dashed] (m) edge[bend left=15] (mc) ;
      
      \draw[->,line width =1pt] (ma) edge  (mb) ; 
      \draw[->,line width =1pt, dashed] (mc) edge [bend right = 35] (ma) ;
      \draw[->,line width =1pt] (mc) edge  (mb) ;
      
    \end{tikzpicture}
  \end{center}
\end{subfigure}
\vspace{-3em}
\begin{center}
  \scalebox{1.0}{
    \hspace{-2em}
    \begin{tabular}{| l || c  | c| c| c|}
      \hline
      step & total order & %pretotal order  &
                                              forcing strategy for $(M,m)$ & values of \ran-nodes & values of \m-control nodes\\ \hline
      
      0  &  $[r_3r_1r_2]$	& %$[(ab)+(ca)]$	&
                                                  $(r_2,r_3)$  & $ .1, .5, .18$	& $.1,.5, .1$ \\
      1  &  $[r_1r_3r_2]$	& %$[(ab)+(ac)+(cb)]$&	
                                               $(r_2,r_3)$ &  $.46,  .5, .54$	& $.46,.5,.5$ \\
      2  &  $[r_1r_2r_3]$	& %$[(ab)+(ac)+(bc)]$&
                                               $(r_2,M)$ &  $.46, .5, .54$	& $.46,.5,.54$ \\
      \hline
    \end{tabular}
  }
\end{center}

 \caption{\\Fig.~(\ref{fig.a}): example of an SSG taken
   from~\cite{gimbert2008simple}. Node $m$ (resp. $M$) belongs to player \m
   (resp. player \M). Note that we add the dummy \ran-node $r_2$ so that the
   game is in CF.\\
  Fig.~(\ref{fig.b}): modified graph obtained by adding the \m-control nodes $1,
  2, 3$ before each \ran-node.\\
  Fig.~(\ref{fig.c}): solving game $G[t]$ with total order $t=[r_3r_1r_2]$. Arcs
  between control nodes are added according to $t$. Forcing strategies for $m$
  and $M$, and  optimal strategy for nodes $m_i$ are shown with dashed edge.
  Hence, the forcing set of $1$ is $\{1\}$, for $2$ it is $\{2, M  \}$,
  and $\{3, m  \}$ for $3$.\\
  Table:  a run of  Algorithm~\ref{alg:iterative_random_ludwig}. The values of
  the \ran-nodes and the \m-control nodes are given in the order from left to right. 
}
\label{fig.example}
\end{figure}

Here are basic properties on $G[t]$ which should explain why we consider this game.

\begin{lemma} \label{lem:Gt}
  \begin{enumerate}[(i)]
  \item if $G$ has stopping \M-strategy, so does $G[t]$;
  \item optimal values $\val_{*,*}(i)$ of control nodes $i \in [1,k]$ in $G[t]$
    are nondecreasing along $t$;
  \item optimal strategies in $G[t]$ coincide with forcing strategies for order $t$
    on $\VM \cup \Vm$;
  \item the game $G[t]$ can be solved in polynomial time.
  \end{enumerate} 
\end{lemma}

\begin{proof}
  To see why $(i)$ is true, just note that since $t$ is an antisymmetric
  relation, this does not create new cycles among $\m$-nodes.

  Suppose now that $(i,j) \in t$. By optimality for the $\m$ player, and since $G[t]$ is
  stopping,
  $\val_{*,*}(i)$ is the minimum value of $\val_{*,*}(x)$ for all outneighbours
  $x$ of $i$ (see Th.~\ref{th:local_optimality_conditions3}).
  Since $j$ is an outneighbour of $i$ in $G[t]$, we have 
  $\val_{*,*}(i) \leq \val_{*,*}(j)$. Hence $(ii)$ is true.

  Now, consider replacing in $G[t]$ every control node $i \in [1,k]$
  by a new sink $s_i$ with value $\val_{*,*}(i)$. Clearly the values of this new game remain the
  same. But, by construction of $G[t]$, random nodes have no incoming
  arcs and they could be as well removed without changing the optimal
  values on $\VM \cup \Vm$. By reducing the game in this way, we get a
  deterministic game whose optimal values on $\VM \cup \Vm$ are the same as those
  of $G[t]$. By definition, optimal strategies of this game are $t$-forcing
  strategies, hence $(iii)$ is true.

  Finally, to solve $G[t]$ we can choose a couple $(\sigma_t,\tau_t)$ of
  $t$-forcing strategies and search for optimal strategies in $G[t]$
  that match with $(\sigma_t,\tau_t)$ on $\VM \cup \Vm$. Hence,
  the strategy of all $\M$-nodes is fixed, and
  only $\m$-strategies on control nodes are computed by solving a one player SSG. It can be done in
  polynomial time by linear programming (see \cite{condon1992complexity}).
\end{proof}

As explained in the proof above, to solve $G[t]$, it is enough to 
compute $t$-forcing strategies on $\VM \cup \Vm$, which can be done in linear time,
and then to solve a one player SSG with only $O(k)$ nodes.

\subsection{Value intervals and pivot} \label{sub:pivot}

In what follows, we write $\val[t]$ for the vector of optimal values of $G[t]$.

\begin{definition}[Constrained control node] \label{def:constrained1}
  We say that a control node $i \in [1,k]$ is {\it constrained} in $G[t]$
  if $\val[t](i) < \val[t](r_i)$.
\end{definition} 

Constrained control nodes are similar to switchable nodes in SSG.
In fact, we can characterize optimality of an order by the absence of constrained node as follows.

\begin{lemma}[Optimal order]\label{lem:opt_t}
  Let $t \in \total(t)$. The game $G[t]$ does not have any constrained control nodes if and only if the forcing strategies $(\sigma_t, \tau_t)$ are optimal strategies for $G$. In this case we say that $t$ is an optimal order for $G$.
\end{lemma}
\begin{proof}
  First note that since $G$ is in CF, $\sigma_t$ is always stopping.
  
  If $G[t]$ does not have any constrained control nodes, then optimal
  strategies are the forcing strategies $(\sigma_t, \tau_t)$ on
  $\VM \cup \Vm$, together with the choice $(i, r_i)$ for each control
  node $i \in [1,k]$. Then, by merging the control nodes with their
  associated random node while removing the unused arcs between the
  control nodes (hence recovering the initial game G), the values on
  the remaining nodes are kept, and so are the optimality conditions
  of Th.~\ref{th:local_optimality_conditions3}.

  If $(\sigma_t, \tau_t)$ are optimal strategies for $G$, then the
  values $v_1,v_2,\dots, v_k$ of the \ran-nodes are nondecreasing
  along order $t$. Hence, by turning $G$ into $G[t]$ and extending
  strategies $(\sigma_t, \tau_t)$ with the choice $(i, r_i)$ for each
  control node $i \in [1,k]$, we will obtain values that satisfy
  optimality conditions and such that $Val[t](i) = Val[t](r_i)$,
  showing that $i$ is not constrained.
\end{proof}

We define the {\it value interval} of a control node $i \in [1,k]$ as the set of $j \in [1,k]$
that share the same optimal value in $G[t]$, i.e. $\val[t](i) =
\val[t](j)$.
This set is indeed an interval in order $t$ by $(ii)$ of Lemma~\ref{lem:Gt}, i.e. its elements are consecutive in order $t$.

\begin{definition} \label{def:pivot}
  The {\it pivot} operation on a control node $i \in [1,k]$ for the order $t$ is the transformation
  of $t$ into a new order $t' \in \total(k)$, obtained by moving $i$ just after
  the end of its value interval in $t$. 
\end{definition} 

Note that if $i$ is the last node of its value interval, then the pivot
operation
does nothing. Also note that if $i$ is constrained, it cannot be the last
node of its value interval (we shall only pivot
on constrained control nodes).

\textbf{Example.} Let $k=7$ and let $t$ be in ascending order $[7,2,4,1,3,6,5]$.
Suppose that the values of control nodes are, in this order $[0.2, 0.2, 0.3, 0.3, 0.3, 0.4, 0.4]$.
% \[
%   \begin{array}{c|ccccccc}
%     t & 7 & 2 & 4 & 1 & 3 & 6 & 5\\
%     \hline
%     \val & 0.2 & 0.2 & 0.3 & 0.3 & 0.3 & 0.4 & 0.4 \\
%   \end{array} 
% \]
The value intervals are $[7,2]$, $[4,1,3]$ and  $[6,5]$. The pivot operation on $4$ places $4$ after
$3$, 
so that the obtained order would be $[7,2,1,3,4,6,5]$. 

The following theorem shows that the pivot operation increases the value vector, 
which will enable us to design a strategy improvement algorithm on the forcing
strategies (where the improvement is on $\val[t]$ rather than on values in the original
game $G$).
A similar theorem is proved in~\cite{dai2009new} to build a different strategy improvement algorithm.

\begin{theorem} \label{thm:pivot}
  Let $t \in \total(t)$ and $i \in [1,k]$ be a constrained control node.
  If $t'$ is obtained from $t \in \total(k)$ by pivoting on $i$, then $\val[t'] > \val[t]$.
\end{theorem} 
% \david{expliquer que $G[t+t']$ est non stopping mais que $G[t']$ l'est donc pas
%   de pb pour le switch. Ou mentioner ce fait dans le thm de la partie switch ?}

\begin{proof}
  Consider a new game $G[t+t']$ which is obtained from $G$ like $G[t]$
  and $G[t']$ but with arcs $(i,j)$ for $i \neq j$ between control
  nodes for all $(i,j) \in t \cup t'$.  Let $(\sigma,\tau)$ and
  $(\sigma',\tau')$ be respective optimal strategies in $G[t]$ and
  $G[t']$. We can interpret these strategies as strategies in
  $G[t+t']$. Since the only difference between $G[t]$, $G[t']$ and
  $G[t+t']$ are the arcs between control nodes, all strategies
  $(\sigma,\tau)$ give exactly the same values in $G[t]$ and in
  $G[t+t']$, and a similar observation can be made for $G[t']$. Hence, to
  prove the result, is is enough to show that $\val_{\sigma',\tau'}~>
  \val_{\sigma,\tau}$ in $G[t+t']$. Note that whereas
  $\sigma$ and $\sigma'$ are respective stopping \M-strategies of $G[t]$
  and $G[t']$, they could be not stopping in $G[t+t']$. However it is not
  difficult to see that conclusion of
  Th.~\ref{th:switch_augmentation2} would still apply. 
  Hence it is sufficient to show in $G[t+t']$ that changing
  $(\sigma,\tau)$ into $(\sigma',\tau')$ makes a nondecreasing switch
  on every node, and is increasing in at least one node.

  In the order $t$, let $I = [i, i_1, i_2, \cdots, i_\ell]$, with
  $\ell \geq 1$ be the increasing sequence of consecutive nodes
  sharing the same value as $i$ for $(\sigma,\tau)$ (i.e. the value
  interval of $i$, starting from $i$). Since $i$ is constrained,
  $\ell \geq 1$.

  The pivot operation transforms this part of $t$ into
  $[i_1, i_2, \cdots,i_\ell,i]$, hence the only differences between
  $G[t]$ and $G[t']$ are the $\ell$ arcs $(i,i_c)$ for
  $1 \leq c \leq \ell$ that are inverted into $(i_c,i)$.
  %\david{expliciter mieux le truc des forcing sets} 
  Hence, if $j \not\in I$, it keeps the same position relatively to all other control nodes when
  we change the order $t$ into $t'$, hence $\frc[t](j) =
  \frc[t'](j)$.
  Hence, when we change $(\sigma,\tau)$ to $(\sigma',\tau')$, either
  there is no switch in $\frc[t](j)$, or it is between nodes of the
  same values.

  % Hence, if $x$ is a $\M$ (resp $\m$) node in $\frc[t](j)$ for $j \not\in I$,
  % then $\sigma'(x)$ (resp. $\tau'(x)$) is in $\frc[t'](j) \cup \{j\} = \frc[t](j) \cup \{j\}$,
  % so $\val[t](\sigma'(x)) = \val[t](x) = \val[t](j)$ (resp. $\val[t](\tau'(x)) = \val[t](x) =
  % \val[t](j)$).
  % On these nodes, when we change $(\sigma,\tau)$ to $(\sigma',\tau')$, either
  % there is no switch, or it is between vertices of the same
  % values.

  For nodes $x \in \frc[t](j)$ with $j \in I$, clearly
  $\sigma'(x)$ (resp. $\tau'(x)$) is in some $\frc[t](j')$ with $j' \in I$. Since
  all these nodes also share the same value (by definition of the value interval),
  these switches are also between nodes of same values.

  Suppose that there is a decreasing switch on a control node, i.e. for a $j \in [1,k]$ we have $\val[t](\tau'(j)) <
  \val[t](j)$. In this case $\tau'(j)$ should be stricly before $j$ in $t$ since
  optimal values are increasing along $t$. So we could not have $(j,\sigma'(j))
  \in t$ but should have $(j,\sigma'(j)) \in t'$. The only possibility
  is  $\sigma'(j) = i$ and $j \in I$. Since these nodes are in the same
  value interval, once again this switch is unchanging, a contradiction.

  We showed that no switch from $(\sigma,\tau)$ to $(\sigma',\tau')$ is
  decreasing.
  Now consider the case of $i$ during the pivot operation. 
  Since $i$ is constrained, $\val[t](i) < \val[t](r_i)$.
  Since $\tau'(i)$ can either be equal to $r_i$
  or to some $j$ which is striclty after the value interval of $i$ in order $t$,
  hence has a greater value, we see that the switch at $i$
  must be increasing.
%   This concludes the proof.
\end{proof}

\subsection{Main algorithm}

Algorithm~\ref{alg:iterative_random_ludwig} consists in iterating on orders $t \in \total(k)$, by picking
randomly a pivotable element in $t$ and updating $t$ by a pivot on $i$, until
we reach an optimal order.

Here is the \textbf{pivot selection rule}. First, prior to the
execution of the algorithm, we choose randomly and uniformly an order
$\Theta$ on the set of all $\frac{k(k-1)}2$ unordered pairs of control
nodes $\{i,j\}$, with $i,j \in [1,k]$.  Then, at each step of the
algorithm, consider the game $G[t]$, and remove one by one the arcs
between control nodes, following order $\Theta$. During this process,
choose as pivot the first constrained control node, if any, which is
disconnected from the following nodes of its value interval.  In more
detail, for a given order $t$, compute $\val[t]$ and then partition
the control nodes into value intervals. Each constrained control node
$i$ has $d(i)$ arcs leading to other control nodes from the same value
interval, where $d(i)$ is its distance in $t$ to the last element of
this interval. Enumerating $\Theta$ in ascending order, the pivot is
the first constrained node $i$ whose $d(i)$ arcs are encountered.

\textbf{Example.} Continued from the previous example with $k=7$ and value intervals $[7, 2]$, $[4, 1,
3]$ and $[6, 5]$. Suppose that the order $\Theta$ starts
$  \{2,5\} ,\{7,6\} ,  \{1,4\}, \{ 2,7 \}\dots.$
The first element that is disconnected from its value interval is $7$ which is the one we choose as a pivot leading to order $[2,7,4,1,3,6,5]$.

\begin{algorithm}
\DontPrintSemicolon
  \SetKwInOut{Input}{input}\SetKwInOut{Output}{output}
  \Input{$G$ SSG, initial total order $t \in \total(k)$.}
  \Output{optimal $\M$ and $\m$-strategies.}
    \BlankLine
    $\cdot$  pick randomly and uniformly a total order $\Theta$ on all
    pairs of control nodes  $\{i,j\}$\;% with $i,j \in [1,k]$\;
    $\cdot$ compute values in $G[t]$ (poly. time, see Lemma \ref{lem:Gt})\;
    \While{$t$ is not optimal (Lemma \ref{lem:opt_t})}{
      $\cdot$ choose a pivot $i$ by the pivot selection rule \;
      $\cdot$ update $t$ by pivoting on $i$ as in Def. \ref{def:pivot}\;
    $\cdot$ compute values in $G[t]$ (poly. time, see Lemma \ref{lem:Gt})\;
    }
    \Return{a couple of forcing strategies for order $t$ in $G$.}
  \caption{Iterative version for Bland's rule on random nodes} \label{alg:iterative_random_ludwig}
\end{algorithm}

By Th.~\ref{thm:pivot}, no order $t \in \total(k)$ is repeated during the
execution of Algorithm \ref{alg:iterative_random_ludwig}; since $\total(k)$ is finite, the algorithm reaches in a finite number of steps an order $t^* \in \total(k)$ which has no
constrained node, i.e. which is optimal by Lemma \ref{lem:opt_t}.
Hence, Algorithm \ref{alg:iterative_random_ludwig} computes 
optimal strategies for $G$ in at most $k!$ steps. However, we claim the following result, which will be proved in the next section.

\begin{theorem} \label{thm:mainthm}
  Alg.~\ref{alg:iterative_random_ludwig} computes optimal strategies for $G$
  in at most $e^{\sqrt2 \cdot k}$ expected steps. 
\end{theorem}

Note that for $k$ large enough we have $e^{\sqrt2 \cdot k} < k!$, whose growth
is roughly equivalent to $2^{k \log k}$. Moreover, the algorithm uses $O(k^2 \log k)$ random bits to choose the order $\Theta$ on pairs.

\section{Analysis of Algorithm \ref{alg:iterative_random_ludwig}} \label{sec:analysis}

In this section we prove Theorem \ref{thm:mainthm}.  To do this we
shall reformulate Algorithm \ref{alg:iterative_random_ludwig} as a
recursive algorithm, but we need additional notions for this.  The
recursive formulation also reveals the nature of the algorithm: it
computes an optimal order on  control nodes by finding the right
order between each pair of these nodes using dichotomy. This
allows the same analysis as for Ludwig's Algorithm and its variants.

\subsection{Modified game $G[p]$ for a pretotal order $p$}

If $p \in \pretotal(k)$ is a pretotal order, we define $G[p]$ exactly as was defined $G[t]$ for a
total order $t \in \total(k)$ in section \ref{sub:modified}.
The only difference is that, since $p$ is not total, a control node
$i \in [1,k]$ only has arcs to those $j \neq i \in [1,k]$ such that
$(i,j) \in p$. 
% In particular, if $p$ is an `empty' order, i.e. $p = \{(i,i) : i \in [1,k]\}$,
% then $G[p]$ is isomorphic to $G$.

To simplify notation, for any node $x$ in $G[p]$, define $\val_*[p](x) := \val_{*,*}^{G[p]}(x)$ as the optimal value of $x$ in $G[p]$.
We can now directly extend some of the observations of Lemma \ref{lem:Gt} to pretotal orders.

\begin{lemma} \label{lem:mimj}
  If $p \in \pretotal(k)$, then optimal values of control nodes $i \in [1,k]$ in $G[p]$
    are nondecreasing in order $p$, i.e. if $(i,j) \in p$
    then $\val_*[p](i) \leq \val_*[p](j)$.
\end{lemma}

In order to solve $G[p]$, the algorithm will recursively compute an
optimal total ordering of control nodes $i \in [1,k]$ extending $p$. 
% Indeed, if we
% know the order of the values $\val_*[p](i)$ for $i \in [1,k]$, then we
% can easily solve $G[p]$ by linear programming as was explained in the
% proof of Lemma \ref{lem:Gt}.
% Furthermore, if $t \in \total(k)$ is a  nondecreasing ordering
% of the values $\val_*[p](i)$, then, by Lemma~\ref{lem:mimj}, $t$ extends $p$
% as an order. We call such a $t$ an {\it optimal order} for $G[p]$, and we will
% prove a characterization of optimal orders in the next Section.
Thus, for all total orders $t \in \total(k)$ extending $p \in \pretotal(k)$,
we need to assign a  value in $G[p]$, which we denote $\val[p](t)$. Here is how
we define it.

\begin{definition}
\label{def:partialvalues}
  Let $t \in \total(k)$ extending $p \in \pretotal(k)$. The values
  $\val[p](t)$ associated to $t$ in $G[p]$ are the values
  $\val_{\sigma_t,\tau_t}$ where $\sigma_t$ and $\tau_t$ satisfy:
  \begin{enumerate}[(i)]
  \item $\sigma_t$ and $\tau_t$ are forcing strategies for $G[t]$;
  \item $\tau_t$ statisfies the \m-optimality conditions (Thm. \ref{th:local_optimality_conditions3}) on every control node $i \in [1,k]$.
  \end{enumerate} 
  % \begin{enumerate}[(i)]
  % \item As in XXX, the SDG associated to $t$ 
  %   is the SDG obtained from $G$ by replacing
  %   all random nodes $r_i$ of $G$, $i \in [1,k]$, by a sink with value $v_i$,
  %   such that $t$ i a total ordering of the numbers $v_i, i \in [1,k]$.
  %   We can identify strategies in this $SDG$ and strategies in $G$ since the
  %   $\M$ and $\m$ nodes are the same.

  % \item A $\M$-strategy $\sigma$ (resp. $\m$-strategy $\tau$) in $G[p]$ is called
  %   a forcing strategy for $t$ if the projection of $\sigma$
  %   (resp. $\tau$) onto $\VM$ (resp. $\Vm$) is an optimal
  %   strategy in the SDG associated to $t$.

  % \item The values $\val[p](t)$ associated to $t$ in $G[p]$ are the values
  %   $\val_{\sigma_t,\tau_t}$ where $\sigma_t$ and $\tau_t$  satisfy:
  %   \begin{enumerate}
  %   \item $\sigma_t$ and $\tau_t$ are forcing for $t$;
  %   \item $\tau_t$ statisfies the optimality conditions XXX on every control node $i \in [1,k]$.
  %   \end{enumerate} 
  % \end{enumerate} 
\end{definition} 

% Let us comment briefly on the previous definition. The values $\val[p](t)$
% correspond
% to a situation  where players try to force from $\VM \cup \Vm$ to access to a
% control node
% $i \in [1,k]$ as high in $t$ as possible for $\M$, and as low in $t$ as possible
% for $\m$.
% To ensure this, any couple of optimal strategy $(\sigma_t,\tau_t)$ from the SDG associated to $t$
% will do.
% Then, once these are fixed, remains to decide for $\m$ how to play on each
% control node $i \in
% [1,k]$, and
% this can be decided by solving the game one-player game where strategies are
% fixed
% on $\VM \cup \Vm$.

% Now we can characterize optimal strategies:

As a summary, $\val_*[p]$ is the vector of optimal values of game $G[p]$ while
$\val[p](t)$ is the vector of optimal values of $G[p]$ when the strategies in
$\Vm$ and $\VM$ are forcing strategies. It follows that $\val[p](t) \leq
\val_*[p]$. Recall that $\val[t]$ is the vector of optimal values of $G[t]$. Then we have $\val[t] = \val_*[t] = \val[t](t)$.

% 
% \subsubsection{Optimal total orders}
% 

% If $p \in \pretotal(k)$, define $\val_*[p]$ as the optimal value of $G[p]$, i.e.
% \[\val_*[p] := \val_{*,*} G[p].\]

\begin{definition}[Optimal order]
  Let $p \in \pretotal(k)$ and $t \in \total(k)$ extending $p$
  ($ p \subset t$). We say that $t$ is an optimal total order for $p$
  if $\val_*[p] = \val[p](t)$.
\end{definition}

The next lemma proves the existence and gives a characterization of optimal
orders.

\begin{lemma}%[See Appendix~\ref{app:proof}]
  \label{lem:carac_optimal_order}
  Suppose $G$ is in CF and let $p\in \pretotal(k), t \in \total(k), p \subset
  t$. Then the following conditions are equivalent:
  \begin{enumerate}[(i)]
  \item $t$ is an optimal order for $p$;
  \item $t$ is a nondecreasing ordering of the values $\val_*[p](i)$ for $i \in [1,k]$;
  \item $\val_*[p] = \val[t]$.
  \end{enumerate}
\end{lemma}

\begin{proof}[Proof of Lemma~\ref{lem:carac_optimal_order}]

  First, note that, by definition, optimality
  conditions are satisfied at control nodes in the definition
  of $\val[p](t)$, so it is always true that values $\val[p](t)(i)$
  are nondecreasing along $p$.

  Suppose now that $t$ is optimal for $p$, i.e.  $\val_*[p] = \val[p](t)$,
  and suppose that $t$ is not a nondecreasing ordering
  of the values  $\val_*[p](i)$ for $i \in [1,k]$.
  Then there must be two consecutive $i,j$ in order
  $t$ such that $\val_{*}[p](i) > \val_{*}[p](j)$,
  and we must have $(i,j) \in t \setminus p$.
  Consider the order $t'$ that we obtain
  from $t$ by inverting $j$ and $i$. Clearly,
  this order also extends $p$ since the only inversion
  between $t$ and $t'$ is $(i,j)$.
  By an argument similar to the proof
  of Theorem \ref{thm:pivot}, it is easy to obtain that
  $\val[p](t') > \val[p](t)$, which contradicts optimality.
  Hence we proved $(i) \Rightarrow (ii)$.

  Now suppose $(ii)$. Let $(\sigma^*,\tau^*)$ be an optimal strategy of
  $G[p]$; we show that optimality conditions
are met in $G[t]$. First note that $(\sigma^*,\tau^*)$ has the
  same values in $G[p]$ and $G[t]$. For the nodes in $\VM \cup \Vm$,
  no new arcs are added so the optimality conditions are still
  satisfied. Now, consider control nodes. An arc
  $(i,j)\in t \backslash p$ cannot lead to a
  lower value for $i$ by assumption $(ii)$. Hence the optimality
  conditions are still satisfied on control nodes, strategies
  $(\sigma^*,\tau^*)$ are optimal for $G[t]$, so finally
  $\val_*[p] = \val[t]$ and we proved $(ii) \Rightarrow (iii)$.

  Since $t$ involves more arcs than $p$ between the control \m-nodes, we have $\val[t] \leq \val[p](t) \leq \val_*[p]$.
  Assume $(iii)$, then $\val_*[p] = \val[p](t)$. Hence we proved $(iii) \Rightarrow (i)$.
  % and consider a couple $(\sigma_t,\tau_t)$ of forcing strategies corresponding to
  % $t$ as in definition \ref{def:partialvalues}. By definition, optimality
  % conditions are satisfied at every control node $i \in [1,k]$.
  % Moreover, by lemma XXX (TODO A FAIRE DANS LA PARTIE SDG), since
  % $t$ is a nondecreasing order for values of the control nodes,
  % optimality conditions are satisfied also on $\VM \cup \Vm$. All in
  % all, since $G$ is stopping, $(\sigma_t,\tau_t)$ are optimal in $G[p]$
  % hence $t$ is optimal for $p$, so $(i)$ is true.
\end{proof}

\subsection{Recursive formulation}

We now give Algorithm~\ref{alg:recursive_random_ludwig}, a recursive formulation of Algorithm~\ref{alg:iterative_random_ludwig}.
We will prove that these two algorithms compute exactly the same sequence of
total orders, and use the recursive formulation to derive a bound.

\begin{algorithm}
\DontPrintSemicolon
  \SetKwInOut{Input}{input}\SetKwInOut{Output}{output}
  \Input{$G$ SSG, order $\Theta$ on all pairs $\{i,j\}$ with $i\neq j \in
    [1,k]$, initial total order $t_0$ in $\total(k)$ extending a pretotal order $p_0 \in \pretotal(k)$.}
  \Output{optimal total order of $G[p_0]$}
    \BlankLine
    \lIf{$p_0$ is total (i.e. $p_0 = t_0$)}{\Return{$t_0$}}
    
    \Else{
      $\cdot$ select according to $\Theta$ the last pair $\{i,j\}$ s.t. $(i,j) \in t_0\setminus p_0$ \;
      $\cdot$ let $p_1 = p_0 + (i,j)$ and $p_2 = p_0
      +(j,i) $\;
      $\cdot$ recursively solve $G[p_1]$ with initial order $t_0$,
      giving optimal total order $t_1$ for $G[p_1]$ \;
      \lIf{$t_1$ is optimal for $G[p_0]$ (apply criterium in Lemma
        \ref{lem:carac_optimal_order})}{\Return $t_1$}
      
      \Else{
        $\cdot$ let $t_2$ be the total order obtained by 
        pivoting in $t_1$ along $i$\;
        $\cdot$ recursively solve $G[p_2]$ with initial order $t_2$, giving
        optimal order $t^*$ for $G[p_1]$\;
        \Return $t^*$.\
      }
    }
  \caption{Recursive version for Bland's rule on random nodes} \label{alg:recursive_random_ludwig}
\end{algorithm} 
% 
% \subsection{Equivalence between Algorithms \ref{alg:iterative_random_ludwig} and
%   \ref{alg:recursive_random_ludwig}}

% To establish the equivalence between both algorithms, we need additional properties on the total optimal order.

\begin{definition}
  Let $p_0 \in \pretotal(k)$ and $\{i,j\} \notin p_0$ such that $p_1 =  p_0 + (i,j)$ is still a pretotal order.
  We say that the addition of $(i,j)$ to $p_0$ is {\it constraining},
  or that $(i,j)$ is {\it constrained},
  if $\val_*[ p_1 ] < \val_*[p_0]$.
\end{definition}

When an arc is constrained, it is essential to the
$\m$-optimal strategy in $G[p_1]$; in other words removing
this arc would increase optimal values.

\begin{lemma}%[See Appendix~\ref{app:proof}]
  \label{lem:property_optimal_order}
  Suppose $G$ is in CF and let $p\in \pretotal(k)$, $t \in \total(k)$, $p \subset
  t$. Then the following conditions are equivalent:
  \begin{enumerate}[(i)]
  \item $t$ is an optimal ordering for $p$;
  \item the addition of every arc $(i,j) \in t \setminus p$ to $p$  is not constraining;
  \end{enumerate}
\end{lemma}

\begin{proof}[Proof of Lemma~\ref{lem:property_optimal_order}]
  Let $p_1 =  p_0 + (i,j )$.   
  Since $p \subset p_1 \subset t$, we have
  $\val[t]\leq \val_*[p_1] \leq\val_*[p]$.

  Assume that $t$ is an optimal order for $p$. If the addition of $(i,j)$
  to $p$ is constraining, then $\val_*[t]\leq \val_*[p_1] < \val_*[p]$
  which contradicts $\val[t] = \val_*[p]$ by Lemma
  \ref{lem:carac_optimal_order}. Hence we proved
  $(i) \Rightarrow (ii)$.

  Assume that no arc in $t \backslash p$ is constraining. Then, add
  sequentially arcs in $t \backslash p$ to $G[p]$ until we get $G[t]$,
  hence forming a sequence of games
  $G[p] = G[p_0], G[p_1], \dots, G[t]$. If none of these arcs is used
  in $G[t]$ then $\val[t] = \val_*[p]$ and $(i)$ is proved. Otherwise
  consider the first arc $(i,j)$ such that
  $\val_*[p_\ell+(i,j)] < \val_*[p_\ell]$.  This implies that $\val_*[p_{\ell}](j) <
  \val_*[p_{\ell}](i)$.
  But $\val_*[p] = \val_*[p_{\ell}]$ since no constraining arc has
  been added until step $\ell$. Hence $\val_*[p](j) < \val_*[p](i)$,
  and finally $(i,j)$ is constraining for $p$, a contradiction.
\end{proof}

Consider a run of the recursive algorithm and let $t$ be a total order
at any step of the run. Let us inspect the first time where $t$ is
modified. Order $t$ will be optimal for a sequence of pretotal orders
that are obtained from $t$ by removing one by one pairs in order
$\Theta$ as long as they are not constrained. This in fact amounts to
ascending the recursive call tree. Let $(i,j)$ be the first
constrained pair and $p_0$ the pretotal order obtained once $(i,j)$ is
removed. Then $t$ is turned into $t'$ by pivoting on node $i$ as we
did in iterative Alg.~\ref{alg:iterative_random_ludwig}.  We show now
that control node $i$ is same as the pivot selected by the pivot
selection rule.

Note first that, during the process of removing the pairs one by one
in order $\Theta$, the value intervals of $G[t]$ are kept unchanged until the
pretotal order $p_0+(i,j)$ is reached. Since removing $(i,j)$ implies
an increase of the optimal values, it means that $i$ and $j$ were in
the same value interval and that $i$ had no other neighbour in that
interval. Note here that, as a consequence, $p_0+(j,i)$ is then
guaranteed to be a pretotal order. Clearly, node $i$ is the first
control node in that situation. So the choice of node $i$ exactly
obeys the pivot selection rule.

Finally, the following lemma enables us to analyze the complexity of Algorithm~\ref{alg:recursive_random_ludwig}.

\begin{lemma} \label{lem:fucktrump}
  Let $p_0 \in \pretotal(k)$ and $\{i,j\} \notin p_0$ such that $p_1 =  p_0 + (i,j) $ and   $p_2 =  p_0 + (j,i)$  are pretotal orders and where the addition of $(i,j)$ to $p_0$ is constraining.
  Let $t_1$ be an optimal total order for $p_1$, $t_2$
  obtained from $t_1$ by pivoting in $i$,
  and let $t^*_2$ be an optimal total order for $p_2$.
  
  Let $(i_1,j_1)$ such that
  $ \val_*[ p_0 + (i,j) ] \not\leq
  \val_*[ p_0 + (i_1,j_1)]. $
  Then for any total order $t$ obtained by Algorithm \ref{alg:recursive_random_ludwig}   between
  $t_1$ and $t_2^*$ (including those), one has $(j_1,i_1) \in t$.
\end{lemma}

\begin{proof}
  Suppose that $(i_1,j_1) \in t$. Then $t \supset   p_0 + (i_1,j_1) $
  hence  $\val[t] \leq \val_*[ p_0 + (i_1,j_1) ] $.

  On the other hand, since the pivot operation is increasing values, we
  have $\val_*[p_0 + (i,j)] = \val[t_1] < \val[t_2] \leq \val[t]$,
  so $\val_*[p_0+(i,j)] \leq \val_*[p_0+(i_1,j_1)]$, a contradiction.
\end{proof}

Using this result, the proof for the complexity bound is the same as the proof of
Theorem \ref{thm:ludwigglobal} using the recursive formulation.
Let $f^\Theta(G,p_0,t_0)$ be the total number of pivots performed by Algorithm
\ref{alg:recursive_random_ludwig} on input $G, p_0, t_0$ for an order
$\Theta$ on pairs. 

Now define $\Phi(m) = \sup_{G,p_0,t_0} \espsb{\Theta}{}{ f^\Theta(G,p_0, t_0})$
where the supremum is taken over all games $G$, pretotal orders $p_0$
and total orders $t_0$ extending $p_0$ such that $t_0 \setminus p_0$
is of size at most $m$. The expectation is taken over all possible
uniform choices for $\Theta$.

Then by Lemma \ref{lem:fucktrump}, $\Phi(m)$ will satisfy Lemma \ref{lem:ineq1}, hence the claimed bound of Th.~\ref{thm:mainthm} by Lemma \ref{lem:borneF} since the depth of the recursive tree is at most $\frac{k(k-1)}{2}$.

% \section{A inclure}

% \subsection{Example}

% We illustrate the algorithm on an example from Gimbert and Horn article, shown on Figure~\ref{fig.example}. On Figure~\ref{fig.example2}, MIN nodes are added before the random nodes.

% Let us run the algorithm: we choose an arbitrary order over the set of
% pairs of random nodes, say the order $((ab),(ac),(bc))$. Now let us
% start with the total order on the random nodes $t_0=(cab)$. The graph
% corresponding to the game $G[t_0]$ is shown on next
% Figure~\ref{fig.example3}. The forcing strategy for Min and Max is to
% choose the edge on the right. Then, one can compute the optimal
% answer for the Min nodes $m_a, m_b$, and $m_c$: $m_c$ goes to $m_a$,
% $m_a$ goes to $a$, and $m_b$ goes to $b$. It follows that the value
% for $c,a,b$ are $0.18, 0.1, 0.5$ respectively, and the values for
% $m_c,m_a,m_b$ are $0.1, 0.1, 0.5$ which are increasing according to
% the total order $t_0$ as expected. Then, the pair $(ca)$ is
% constrained (the other possible value for $m_c$ is that of $c$ (0.18)
% or $m_b$ (0.5) which are strictly greater than $0.1$). So we switch
% this pair and obtain the new order $t_1 = (acb)$. The next steps are
% summarized in the following table~\ref{tab.example}.

%% 
%% Bibliography
%% 

%% Please use bibtex, 

\bibliography{ludwig.bib}

\begin{thebibliography}{10}

\bibitem{andersson2008deterministic}
Daniel Andersson, Kristoffer~Arnsfelt Hansen, Peter~Bro Miltersen, and
  Troels~Bjerre S{\o}rensen.
\newblock Deterministic graphical games revisited.
\newblock In {\em Conference on Computability in Europe}, pages 1--10.
  Springer, 2008.

\bibitem{andersson2009complexity}
Daniel Andersson and Peter~Bro Miltersen.
\newblock The complexity of solving stochastic games on graphs.
\newblock In {\em International Symposium on Algorithms and Computation}, pages
  112--121. Springer, 2009.

\bibitem{auger2014finding}
David Auger, Pierre Coucheney, and Yann Strozecki.
\newblock Finding optimal strategies of almost acyclic simple stochastic games.
\newblock In {\em International Conference on Theory and Applications of Models
  of Computation}, pages 67--85. Springer, 2014.

\bibitem{bland1977new}
Robert~G Bland.
\newblock New finite pivoting rules for the simplex method.
\newblock {\em Mathematics of operations Research}, 2(2):103--107, 1977.

\bibitem{calude2017deciding}
Cristian~S Calude, Sanjay Jain, Bakhadyr Khoussainov, Wei Li, and Frank
  Stephan.
\newblock Deciding parity games in quasipolynomial time.
\newblock In {\em Proceedings of the 49th Annual ACM SIGACT Symposium on Theory
  of Computing}, pages 252--263. ACM, 2017.

\bibitem{chatterjee2009termination}
Krishnendu Chatterjee, Luca~de Alfaro, and Thomas~A Henzinger.
\newblock Termination criteria for solving concurrent safety and reachability
  games.
\newblock In {\em Proceedings of the twentieth annual ACM-SIAM symposium on
  Discrete algorithms}, pages 197--206. SIAM, 2009.

\bibitem{DBLP:journals/corr/abs-1106-1232}
Krishnendu Chatterjee and Nathana{\"e}l Fijalkow.
\newblock A reduction from parity games to simple stochastic games.
\newblock In {\em GandALF}, pages 74--86, 2011.

\bibitem{chen2013automatic}
Taolue Chen, Vojt{\v{e}}ch Forejt, Marta Kwiatkowska, David Parker, and Aistis
  Simaitis.
\newblock Automatic verification of competitive stochastic systems.
\newblock {\em Formal Methods in System Design}, 43(1):61--92, 2013.

\bibitem{chen2013synthesis}
Taolue Chen, Marta Kwiatkowska, Aistis Simaitis, and Clemens Wiltsche.
\newblock Synthesis for multi-objective stochastic games: An application to
  autonomous urban driving.
\newblock In {\em International Conference on Quantitative Evaluation of
  Systems}, pages 322--337. Springer, 2013.

\bibitem{condon1990algorithms}
Anne Condon.
\newblock On algorithms for simple stochastic games.
\newblock In {\em Advances in computational complexity theory}, pages 51--72,
  1990.

\bibitem{condon1992complexity}
Anne Condon.
\newblock The complexity of stochastic games.
\newblock {\em Information and Computation}, 96(2):203--224, 1992.

\bibitem{dai2009new}
Decheng Dai and Rong Ge.
\newblock New results on simple stochastic games.
\newblock In {\em International Symposium on Algorithms and Computation}, pages
  1014--1023. Springer, 2009.

\bibitem{etessami2010complexity}
Kousha Etessami and Mihalis Yannakakis.
\newblock On the complexity of nash equilibria and other fixed points.
\newblock {\em SIAM Journal on Computing}, 39(6):2531--2597, 2010.

\bibitem{friedmann2009exponential}
Oliver Friedmann.
\newblock An exponential lower bound for the parity game strategy improvement
  algorithm as we know it.
\newblock In {\em Logic In Computer Science, 2009. LICS'09. 24th Annual IEEE
  Symposium on}, pages 145--156. IEEE, 2009.

\bibitem{gimbert2008simple}
Hugo Gimbert and Florian Horn.
\newblock Simple stochastic games with few random vertices are easy to solve.
\newblock In {\em Foundations of Software Science and Computational
  Structures}, pages 5--19. Springer, 2008.

\bibitem{halman2007simple}
Nir Halman.
\newblock Simple stochastic games, parity games, mean payoff games and
  discounted payoff games are all \uppercase{LP}-type problems.
\newblock {\em Algorithmica}, 49(1):37--50, 2007.

\bibitem{hansen2015improved}
Thomas~Dueholm Hansen and Uri Zwick.
\newblock An improved version of the random-facet pivoting rule for the simplex
  algorithm.
\newblock In {\em Proceedings of the forty-seventh annual ACM symposium on
  Theory of computing}, pages 209--218. ACM, 2015.

\bibitem{hoffman1966nonterminating}
Alan~J Hoffman and Richard~M Karp.
\newblock On nonterminating stochastic games.
\newblock {\em Management Science}, 12(5):359--370, 1966.

\bibitem{ibsen2012solving}
Rasmus Ibsen-Jensen and Peter~Bro Miltersen.
\newblock Solving simple stochastic games with few coin toss positions.
\newblock In {\em European Symposium on Algorithms}, pages 636--647. Springer,
  2012.

\bibitem{kalai1992subexponential}
Gil Kalai.
\newblock A subexponential randomized simplex algorithm.
\newblock In {\em Proceedings of the twenty-fourth annual ACM symposium on
  Theory of computing}, pages 475--482. ACM, 1992.

\bibitem{ludwig1995subexponential}
Walter Ludwig.
\newblock A subexponential randomized algorithm for the simple stochastic game
  problem.
\newblock {\em Information and computation}, 117(1):151--155, 1995.

\bibitem{shapley1953stochastic}
Lloyd~S Shapley.
\newblock Stochastic games.
\newblock {\em Proceedings of the National Academy of Sciences of the United
  States of America}, 39(10):1095, 1953.

\bibitem{stirling1999bisimulation}
Colin Stirling.
\newblock Bisimulation, modal logic and model checking games.
\newblock {\em Logic Journal of IGPL}, 7(1):103--124, 1999.

\bibitem{tripathi2011strategy}
Rahul Tripathi, Elena Valkanova, and VS~Anil Kumar.
\newblock On strategy improvement algorithms for simple stochastic games.
\newblock {\em Journal of Discrete Algorithms}, 9(3):263--278, 2011.

\end{thebibliography}

\end{document}